\documentclass[11pt]{article}

\usepackage[ruled]{algorithm2e}
\usepackage{amsmath}
\usepackage{amsfonts}
\usepackage{amssymb}
\usepackage{amsthm}
\usepackage{bm}
\usepackage{cite}
\usepackage{commath}
\usepackage{dsfont}
\usepackage[shortlabels]{enumitem}
\usepackage{float}
\usepackage{fullpage}
\usepackage[hidelinks]{hyperref}
\usepackage[capitalize,nameinlink,noabbrev]{cleveref} 
\usepackage{indentfirst}
\usepackage[utf8]{inputenc}
\usepackage{mathrsfs}
\usepackage{mathtools}
\usepackage{microtype}
\usepackage{parskip}
\usepackage{scrextend}
\usepackage{setspace}
\usepackage{thmtools,thm-restate}
\usepackage[dvipsnames]{xcolor}
\usepackage{tikz}
\usepackage{tikz-cd}
\usepackage{environ}
\makeatletter
\newsavebox{\measure@tikzpicture}
\NewEnviron{scaletikzpicturetowidth}[1]{%
  \def\tikz@width{#1}%
  \begin{lrbox}{\measure@tikzpicture}%
  \BODY
  \end{lrbox}%
  \pgfmathparse{#1/\wd\measure@tikzpicture}%
  \BODY
}

\newcommand{\defn}{\vcentcolon=}

\newtheorem{theorem}{Theorem}[section]

\newtheoremstyle{mythmstyle}
  {5pt} 
  {\topsep} 
  {} 
  {} 
  {\bfseries} 
  {.} 
  {.5em} 
  {} 
\theoremstyle{mythmstyle} 

\newtheorem{lemma}[theorem]{Lemma}
\newtheorem{corollary}[theorem]{Corollary}

\newtheorem{observation}[theorem]{Observation}

\crefname{proposition}{Proposition}{Propositions}
\crefname{claim}{Claim}{Claims}
\crefname{fact}{Fact}{Facts}

\newtheoremstyle{mydefstyle}
  {5pt} 
  {\topsep} 
  {} 
  {} 
  {\bfseries} 
  {.} 
  {.5em} 
  {} 
\theoremstyle{mydefstyle} 

\newtheorem{definition}[theorem]{Definition}

\newtheorem{construction}[theorem]{Construction}

\theoremstyle{remark}           

\numberwithin{equation}{section}

\definecolor{DarkGreen}{rgb}{0.1,0.5,0.1}
\definecolor{DarkRed}{rgb}{0.5,0.1,0.1}
\definecolor{DarkBlue}{rgb}{0.1,0.1,0.5}
\definecolor{cerulean}{rgb}{0.0, 0.48, 0.65}

\newcommand{\vecc}{\mathbf{c}}

\newcommand{\vece}{\mathbf{e}}

\newcommand{\vecv}{\mathbf{v}}

\newcommand{\vecx}{\mathbf{x}}
\newcommand{\vecy}{\mathbf{y}}

\newcommand{\veczero}{\mathbf{0}}

\newcommand{\cC}{\ensuremath{\mathcal{C}}}

\newcommand{\cM}{\ensuremath{\mathcal{M}}}

\newcommand{\cP}{\ensuremath{\mathcal{P}}}

\newcommand{\mA}{\ensuremath{\mathbf{A}}}

\newcommand{\mD}{\ensuremath{\mathbf{D}}}

\newcommand{\mG}{\ensuremath{\mathbf{G}}}
\newcommand{\mH}{\ensuremath{\mathbf{H}}}

\newcommand{\mM}{\ensuremath{\mathbf{M}}}

\newcommand{\mP}{\ensuremath{\mathbf{P}}}

\newcommand{\mS}{\ensuremath{\mathbf{S}}}

\newcommand{\R}{{\mathbb R}}

\newcommand{\F}{{\mathbb F}}

\newcommand{\N}{\ensuremath{\mathbb{N}}}
\newcommand{\GL}{\ensuremath{\mathcal{GL}}} 
\newcommand{\SP}{\ensuremath{\mathcal{SP}}} 

\newcommand{\class}[1]{\ensuremath{\mathsf{#1}}\xspace}

\newcommand{\CE}{\class{CE}}
\newcommand{\PCE}{\class{PCE}}
\newcommand{\SPCE}{\class{SPCE}}
\newcommand{\LCE}{\class{LCE}}

\newcommand{\LIP}{\class{LIP}}

\newcommand{\NP}{\class{NP}}

\newcommand{\poly}{\mathrm{poly}}

\hypersetup{
    colorlinks,
    linkcolor=cerulean,
    citecolor=cerulean,
    urlcolor=cerulean
}

\usepackage[draft,multiuser,inline,nomargin]{fixme}

\FXRegisterAuthor{a}{ea}{\color{cerulean}Alexandra}
\FXRegisterAuthor{n}{en}{\color{orange}Nikhil}

\title{Reductions Between Code Equivalence Problems}
\author{Mahdi Cheraghchi\thanks{University of Michigan --  Ann Arbor. Email: \texttt{mahdich@umich.edu}}, 
Nikhil Shagrithaya\thanks{University of Michigan --  Ann Arbor. Email: \texttt{nshagri@umich.edu}}, 
Alexandra Veliche\thanks{University of Michigan --  Ann Arbor. Email: \texttt{aveliche@umich.edu}}}

\date{}

\begin{document}

\listoffixmes

\maketitle

\begin{abstract}
    In this paper we present two reductions between variants of the Code Equivalence problem. 
    We give polynomial-time Karp reductions from Permutation Code Equivalence (\PCE) to both Linear Code Equivalence (\LCE) and Signed Permutation Code Equivalence (\SPCE). 
    Along with a Karp reduction from \SPCE to the Lattice Isomorphism Problem (\LIP) shown by Bennett and Win (2024), our second result implies a reduction from \PCE to \LIP.
\end{abstract}

\section{Introduction}
\label{sec:intro}

The \emph{Code Equivalence} (\CE) problem asks if two given codes $\cC_1$ and $\cC_2$ are ``equivalent" in some metric-preserving way; variants of the \CE problem specify the type of equivalence.
\emph{Permutation Code Equivalence} (\PCE) asks if the codes are the same up to permutation of the coordinates of codewords, while \emph{Signed Permutation Code Equivalence} (\SPCE) allows equivalence up to signed permutations.
Still more generally, \emph{Linear Code Equivalence} (\LCE) allows an equivalence up to permutation and multiplication by a (non-zero) constant.
The variants of \CE belong to a larger class of isomorphism problems that ask the following question: Given two objects of the same kind, is there an isomorphism that transforms one object into the other?
Other examples of such problems include Matrix Code Equivalence and Graph Isomorphism.

Besides being interesting problems in their own right, \CE problems have many important applications.
Perhaps most notably, the conditional hardness of \CE variants has been used as a security assumption for several cryptographic schemes proposed to be post-quantum.
These include the seminal McEliece public-key encryption scheme \cite{McEliece78}, a recent NIST post-quantum standardization submission called ``Classic McEliece" \cite{ClassicMcEliece}, and the more recent LESS identification scheme \cite{BBPS21,BBPS22}.

Given the relevance of these problems to cryptography, there has been a considerable amount of work on designing efficient algorithms that solve these problems.
Leon~\cite{Leon82} introduced an algorithm for the search version of \PCE that works well for a large number of codes, but still requires exponential time in the worst case.
The Support Splitting Algorithm developed by Sendrier in \cite{Sendrier00} and extended by Sendrier and Simos in \cite{SS13a} gives an algorithm for linear codes that is efficient for codes with small hull, where the hull of a code is defined by the intersection of the code and its dual.
This algorithm, however, does not work for the case where the dimension of the hull is zero, but this case was later handled by Bardet, Otmani, and Saeed-Taha in \cite{BOS19}.
In the latter paper, the authors reduce the problem of deciding \PCE to a weighted version of Graph Isomorphism, and then use a variant of Babai's algorithm for solving the latter problem \cite{Bab16} to give a quasi-polynomial time algorithm that computes \PCE for the zero hull case.

On the other side of the cryptographic coin, considerable attention has been devoted towards understanding the computational hardness of these problems.
Petrank and Roth in \cite{PR97} showed that Graph Isomorphsim reduces to \PCE, and that \PCE is not \NP-complete unless the polynomial hierarchy collapses.
This was used as evidence for the computational hardness of \PCE until Babai introduced his quasi-polynomial time algorithm for deciding Graph Isomorphism \cite{Bab16}.
Even under the assumption that the polynomial hierarchy does not collapse, however, it is still possible that both Graph Isomorphism and \PCE cannot be decided in polynomial time.

The Matrix Code Equivalence problem is concerned with code equivalence in the rank metric.
Couvreur, Debris-Alazard and Gaborit in \cite{CDG11} give a reduction from \LCE to Matrix Code Equivalence, thus showing that determining code equivalence under the rank metric is at least as hard as determining equivalence under the Hamming metric.

In~\cite{SS13a}, Sendrier and Simos give a reduction from \LCE to \PCE that runs in time polynomial in the blocklength $n$ and alphabet size $q$ of the codes.
To do this, they use the \emph{closure} of a code $\cC \subseteq \F_q^n$, which is the set defined by taking every codeword and multiplying each of its $n$ coordinates with all the non-zero field elements to produce a new vector of length $n(q-1)$. 
Then using the fact that multiplication by any non-zero field element induces a permutation on the elements in $\F_q^*$, they reframe scalar multiplication as a permutation over $\F_q^*$.
In \cite{DG23}, Ducas and Gibbons use a specific form of this closure to prove a reduction from \SPCE to \PCE.
Biasse and Micheli in~\cite{BM23} give a search-to-decision reduction for \PCE, which implies that the decision version of \PCE is at least as hard as the search version.
Bennett and Win in \cite{BW24} give several reductions between \CE variants and other isomorphism problems.
They extend the closure technique from \cite{SS13a} to give a reduction from \LCE to \SPCE, and also give a reduction from \SPCE to the \emph{Lattice Isomorphism Problem} (\LIP).

\LIP is an analogous problem to \CE for lattices.
A \emph{lattice} is an infinite discrete additive subgroup of Euclidean space generated by taking all integer linear combinations of a finite set of linearly independent vectors; these vectors form a \emph{basis} for the lattice.
Two lattices are said to be \emph{isomorphic} if there exists an orthogonal transformation that transforms one lattice (basis) into the other.
\LIP asks if such an isomorphism exists.
In their reduction from \SPCE to \LIP, Bennet and Win use a well-known construction (called Construction A) that lifts a linear code over a prime finite field $\F_p$ to a lattice in $\R^n$.
Because this construction is only well-defined for prime fields, their reduction only works for prime fields.
It remains an open problem to find a reduction from any \CE variant to \LIP for non-prime fields.

\subsection{Results}
\label{subsec:results}

Now we state our main results and place them into context of the prior work described above.
In~\cite{BW24}, the authors prove Karp reductions from \LCE to \PCE and from \LCE to \SPCE.
In our work, we show a reverse Karp reduction from \PCE to \LCE.
\medskip

\begin{theorem}[\PCE reduces to \LCE, informal]
    For linear codes with blocklength $n$ over a field of size $q$, there is a Karp reduction from \emph{\PCE} to \emph{\LCE} that runs in $\poly (n, \log q)$ time.
\end{theorem}

The formal statement is given in~\cref{thm:main-PCE->LCE}.
This result, along with the Karp reduction in~\cite{SS13b, BW24} from \LCE to \PCE running in $\poly(n, q)$ time, implies that the problems \LCE and \PCE are computationally equivalent up to factors in the runtime that are polynomial in $n$ and $q$.

Additionally, the Karp reduction from \LCE to \SPCE detailed in~\cite[Theorem 4.4, Corollary 4.5]{BW24} combines with our result to give a reduction from \PCE to \SPCE running in $\poly(n, q)$ time. 
In situations where $q$ is much larger than $n$ (as is the case for Reed-Solomon codes, for example), it is desirable to have the runtime depend only logarithmically on $q$. 
Note that at least $\log q$ time is required to express a single element of the field, so the dependence on $q$ cannot be smaller than $\log q$.
Our second result is a Karp reduction from \PCE to \SPCE whose runtime depends only logarithmically on $q$.
\medskip

\begin{theorem}[\PCE reduces to \SPCE, informal]
    For linear codes  with blocklength $n$ over a field of size $q$, there is a Karp reduction from \emph{\PCE} to \emph{\SPCE} that runs in $\poly (n, \log q)$ time.
\end{theorem}

The formal statement is given in~\cref{thm:main-PCE->SPCE}.
This result, together with a Karp reduction from \SPCE to \PCE from \cite[Lemma 10]{DG23} running in $\poly(n, \log q)$ time, implies that \PCE and \SPCE are computationally equivalent problems up to factors in the runtime that are polynomial in $n$ and $\log q$.

Finally, by combining the above result with the Karp reduction from \SPCE over prime fields to \LIP from~\cite[Theorem 5.1]{BW24}, we obtain the following corollary.

\begin{corollary}[\PCE reduces to \LIP]
    For any prime $p$, there is a Karp reduction from \emph{\PCE} over a field of order $p$ to \emph{\LIP} that runs in $\poly(n, \log p)$ time.
\end{corollary}

\section{Preliminaries}
\label{sec:prelims}

Let $\N$ denote the set of all positive integers. 
For any $n,m \in \N$ such that $n< m$, we use the notation $[n,m]$ to denote the set of integers $\{n,n+1,\ldots,m\}$.
For $n \in \N$, we use the abbreviated notation $[n]\defn[1,n]$.
For any prime power $q\in \N$, we use $\F_q$ to denote the field of size $q$.
For any commutative ring $R$, we denote its multiplicative subgroup by $R^*$; for a field $\F$, we have $\F^* = \F \setminus \{0\}$.

We use boldface lowercase letters, such as $\vecv$, to denote vectors and boldface uppercase letters, such as $\mA$, to denote matrices.
We use $\vece_i$ to denote the column vector containing a one in the $i$-th position and zeros everywhere else, and $\veczero$ will denote the all zeroes column vector.
For any matrix $\mA$ with $n$ columns and $i \in [n]$, we use $\mA[i]$ to denote the $i$-th column.
We also use $\mA[i,j]$ to denote the entry in the $i$-th row and $j$-th column.
To denote the block submatrix of a matrix $\mA$ spanned by rows $r_i$ through $r_j$ (inclusive) and columns $c_k$ through $c_\ell$ (inclusive), we use the notation $\mA[r_i:r_j, c_k:c_\ell]$.

\subsection{Matrix Groups}
\label{subsec:matrix-groups-intro}

For any field $\F$ and $n\in \N$, we use the following notation for special sets of $n \times n$ matrices over $\F$:
$\GL_n(\F)$ denotes the group of invertible matrices,  $\cP_n(\F)$ denotes the set of permutation matrices, $\SP_n(\F)$ denotes the set of signed permutation matrices, and $\cM_n(\F)$ denotes the set of monomial matrices.
A \emph{permutation} matrix $\mP \in \cP_n(\F)$ contains exactly one $1$ in each row and column and $0$s everywhere else.
A \emph{signed permutation} matrix $\mP \in \SP_n(\F)$ contains exactly one non-zero entry in each row and column, but each of these can be either $1$ or $-1$.
A \emph{monomial} matrix $\mM \in \cM_n(\F)$ contains exactly one non-zero entry in each row and column, but these can take values in $\F^*$; any monomial matrix $\mM$ can be written as $\mM = \mD \mP$ for some diagonal matrix $\mD = \text{diag}(d_1,...,d_n)$, where $d_i \in \F^*$, and permutation matrix $\mP \in \cP_n(\F)$.
Each of these sets of matrices forms a group under matrix multiplication.
Furthermore, these satisfy $\cP_n(\F) \subseteq \SP_n(\F) \subseteq \cM_n(\F) \subseteq \GL_n(\F)$.
If it is clear from context, we do not specify the field for these matrix groups.
For any permutation matrix $\mP \in \SP_n$, we denote the corresponding permutation map over the set of column indices by $\sigma_{\mP} \colon [n] \to [n]$.

The following observation will be useful for our reduction in~\cref{sec:proofs}.

\begin{observation}
    \label{rem:S-bijection}
    Any invertible matrix $\mS \in \GL_k(\F)$ induces a bijective map on $\F^k$.
    In particular, for any $\vecx, \vecy \in \F^k$, we have $\mS\vecx = \mS\vecy$ if and only if $\vecx = \vecy$.
    Then when $\mS$ is multiplied by some matrix $\mA$ of the appropriate dimension, it maps identical (distinct) columns in $\mA$ to identical (distinct) columns in $\mS \mA$.
\end{observation}

\subsection{Codes}
\label{subsec:code-prelims}

A \emph{linear code} is a finite-dimensional vector space over a finite field.
Let $\F_q$ be a finite field of size $q$ for some prime power $q$.
Using the standard notation for linear codes, we say $\cC \subseteq \F_q^n$ is a \emph{$[n,k,d]_q$-code} for some $q, n, k, d \in \N$, where $q$ is the \emph{alphabet size}, $n$ is the \emph{blocklength} given by the length of the codeword vectors, $k$ is the \emph{dimension} of $\cC$ as a linear subspace of $\F_q^n$, and $d$ is the \emph{minimum distance} given by the minimum Hamming weight over all non-zero codewords of $\cC$.
Any linear code $\cC$ can be expressed as the row span of a \emph{generator} matrix $\mG \in \F_q^{k\times n}$.
Note this is not unique (as elementary row operations do not change the span of $\mG$).

\subsection{Equivalence Problems}
\label{subsec:problem-defs}

Two codes $\cC, \cC' \subseteq \F_q^n$ are said to be \emph{permutation equivalent} if there exists a permutation of the coordinates of $\cC$ that gives $\cC'$.
We formalize this notion and define two variants of this problem below.

\begin{definition}[\PCE]
    \label{def:PCE}
    For $k, n \in \N$ and field $\F_q$ of size $q$, the \emph{Permutation Code Equivalence} problem, denoted by \PCE, is the following decision problem:
    Given a pair of generator matrices $\mG, \mH \in \F_q^{k \times n}$, decide whether there exist an invertible matrix $\mS \in \GL_k(\F_q)$ and a permutation matrix $\mP \in \cP_n(\F_q)$ for which $\mS \mG \mP = \mH$.
\end{definition}

If we allow code equivalence up to \emph{signed} permutations, we can consider the following problem.
\begin{definition}[\SPCE]
    \label{def:SPCE}
    For $k, n \in \N$ and field $\F_q$ of size $q$, the \emph{Signed Permutation Code Equivalence} problem, denoted by \SPCE, is the following decision problem:
    Given a pair of generator matrices $\mG, \mH \in \F_q^{k \times n}$, decide whether there exist an invertible matrix $\mS \in \GL_k(\F_q)$ and a signed permutation matrix $\mP \in \SP_n(\F_q)$ for which $\mS \mG \mP = \mH$.
\end{definition}

More generally, we can consider code equivalence up to permutation and multiplication by any non-zero constant.
\begin{definition}[\LCE]
    \label{def:LCE}
    For $k, n \in \N$ and field $\F_q$ of size $q$, the \emph{Linear Code Equivalence} problem, denoted by \LCE, is the following decision problem:
    Given a pair of generator matrices $\mG, \mH \in \F_q^{k \times n}$, decide whether there exist an invertible matrix $\mS \in \GL_k(\F_q)$ and a monomial matrix $\mM \in \cM_n(\F_q)$ for which $\mS \mG \mM = \mH$.
\end{definition}

For brevity, we say that a pair of matrices $(\mG, \mH)$ is in \PCE (or \SPCE, \LCE) if $\mG$ and $\mH$ satisfy the conditions required for the pair to be a YES instance of the problem.

\section{Reductions from \PCE to \LCE and \SPCE}
\label{sec:proofs}

We formally state our results below and prove them in this section.
All our reductions are deterministic Karp reductions, so we do not specify this in our statements hereafter.
\medskip

\begin{theorem}[\PCE reduces to \LCE]
    \label{thm:main-PCE->LCE}
    There is a reduction from \emph{\PCE} to \emph{\LCE} that runs in $\poly(n, \log q)$ time, where $n$ is the blocklength and $q$ is the field size of the input code pair.
\end{theorem}
\medskip

\begin{theorem}[\PCE reduces to \SPCE]
    \label{thm:main-PCE->SPCE}
    There is a reduction from \emph{\PCE} to \emph{\SPCE} that runs in $\poly(n, \log q)$ time, where $n$ is the blocklength and $q$ is the field size of the input code pair.
\end{theorem}

We will use the following two lemmas to justify why certain assumptions about the input can be made without loss of generality.

\begin{lemma}
    \label{lem:col-freq-PCE}
    Let $k,n \in \N$ and $q$ be a prime power.
    For any matrices $\mG, \mH \in \F_q^{k\times n}$, if $(\mG, \mH)$ is in \emph{\PCE}, then no column of $\mG$ appears in $\mG$ more times than a column of $\mH$ appears in $\mH$.
\end{lemma}

\begin{proof}
    By definition, there must be some $\mS \in \GL_k$ and $\mP \in \cP_n$ such that $\mS \mG \mP = \mH$.
    By \cref{rem:S-bijection}, $\mS$ takes identical columns in $\mG$ to identical columns in $\mS \mG$.
    When multiplied with $\mG$ on the right, $\mP$ only permutes the columns of $\mG$ and so does not change the frequency with which any column appears in the matrix $\mS \mG$.
    Therefore for every column $\vecc$ in $S \mG$, its corresponding column in $\mH$ must appear the same number of times as $\vecc$ appears in $\mS \mG$.
\end{proof}

\begin{lemma}
    \label{lem:no-zero-column}
    For any $\mG, \mH \in \F_q^{k \times n}$, let $\overline{\mG}$ and $\overline{\mH}$ be the submatrices of $\mG$ and $\mH$, respectively, obtained by removing all columns equal to $\veczero \in \F_q^k$.
    Then $(\mG, \mH)$ is in \PCE if and only if $(\overline{\mG}, \overline{\mH})$ is in \PCE.
\end{lemma}

This lemma follows from the following trivial reduction:
If $\mG$ and $\mH$ satisfy the \PCE condition, there exist matrices $\mS \in \GL_k$ and $\mP \in \cP_n$ such that $\mS \mG \mP = \mH$.
Because $\mS$ is a linear map, it will always map $\veczero \in \F_q^k$ to itself, so there is an invertible submatrix of $\mS$ and a permutation submatrix of $\mP$ for which $\overline{\mG}$ and $\overline{\mH}$ satisfy the \PCE condition.

As a result of~\cref{lem:no-zero-column} and~\cref{lem:col-freq-PCE}, we can assume without loss of generality that any given input pair of matrices do not contain a zero column and have the same column frequency.

Now we define the construction that will be used to transform the input matrices in our reduction.
\begin{construction}
    \label{con:our-construction}
    Given a $k\times n$ matrix $\mA$ over $\F_q$ with column vectors $\mA[1],\dots,\mA[n] \in \F_q^k$, let $m_\mA$ denote the maximum number of times a column appears in $\mA$.
    Denote $m\defn m_\mA + 1$. 
    We construct the $k\times nm$ matrix $\widehat{\mA}$ by appending $m$ copies of each column vector in order.
    More explicitly, for every $i\in [n]$ and $j\in [m]$, we set $\widehat{\mA}[(i-1)m+j]\defn\mA[i]$.
    Define the $(k+1)\times n$ matrix $\mA_1'$ by appending a row of all ones at the bottom of matrix $\mA$, 
    define the $(k+1)\times nm$ matrix $\mA_2'$ by appending a row of all zeros at the bottom of matrix $\widehat{\mA}$, and 
    define the $(k+1)\times (nm+1)$ matrix $\mA_3'$ by placing ones in every entry of the last row and zeros everywhere else.
    Denoting $n' \defn n + 2nm + 1$, we obtain the final $(k+1) \times n'$ matrix $\mA'$ by concatenating the block matrices $\mA_1'$, $\mA_2'$, and $\mA_3'$:
    \begin{equation*}
        \mA' \defn
        \begin{bmatrix}
            \quad \mA_1' & \mid & \mA_2' & \mid & \mA_3'\quad
        \end{bmatrix} .
    \end{equation*}

    \begin{figure}[H]
        \centering
            \begin{tikzpicture}
              \draw[-] (0,0) -- (8,0) node[right] {};
              \draw[-] (0,1.8) -- (8,1.8) node[right] {};
              \draw[-] (0,0.5) -- (8,0.5) node[right] {};
              
              \draw[-] (0,0) -- (0,1.8) node[right] {};
              \draw[-] (2.1,0) -- (2.1,1.8) node[right] {};
              \draw[-] (4.8,0) -- (4.8,1.8) node[right] {};
              \draw[-] (8,0) -- (8,1.8) node[right] {};
        
              \node[label={$1\ 1\ \dots\ 1$}] at (1,-0.15) {};
              \node[label={$0\ 0\ \dots\dots\ 0$}] at (3.4,-0.15) {};
              \node[label={$1\ 1\ \dots\dots\dots\ 1$}] at (6.4,-0.15) {};
        
              \node[label={$\mA$}] at (1,0.7) {};
              \node[label={$\widehat{\mA}$}] at (3.4,0.7) {};
              \node[label={$\mathbf{0}$}] at (6.4,0.7) {};
            \end{tikzpicture}
        \caption{The matrix $A'$ obtained by \cref{con:our-construction}.}
        \label{fig:matrix-constr}
    \end{figure}

    Note that if the given matrix $\mA$ has full row rank, then $\mA'$ must have full row rank.
\end{construction}
\medskip

With these assumptions and definitions in place, we now prove~\cref{thm:main-PCE->LCE}.


\begin{proof}[Proof of \cref{thm:main-PCE->LCE}:]
    Denote the pair of input matrices by $\mG, \mH \in \F_q^{k\times n}$.
    By~\cref{lem:no-zero-column}, we can assume without loss of generality that $\mG$ and $\mH$ do not contain any zero columns.
    Moreover, we can assume that $\mG$ and $\mH$ have the same row rank, because otherwise the codes represented by these matrices differ in size, and hence cannot be equivalent.
    We can also assume that their row rank is equal to $k$, as otherwise we can consider the submatrices consisting only of linearly independent rows, without loss of generality.
    Let $m_\mG$ denote the maximum number of times a column appears in $\mG$, and define $m_\mH$ similarly.
    By~\cref{lem:col-freq-PCE}, we must have $m_\mG=m_\mH$.
    Denoting $m \defn m_\mG + 1 = m_\mH +1$ and $n' \defn n + 2nm + 1$, we obtain $(k+1) \times n'$ matrices $\mG'$ and $\mH'$, from matrices $\mG$ and $\mH$ respectively, by using~\cref{con:our-construction}.
    Note that because $m\leq n+1$, these matrices can be constructed deterministically in $\poly(n, \log q)$ time.
    The claim in~\cref{thm:main-PCE->LCE} then follows from~\cref{lem:PCE=>PCE} and~\cref{cor:LCE=>PCE}, which are stated and proven below.
\end{proof}

The proof of~\cref{thm:main-PCE->SPCE} is nearly identical to the above proof, but with a restriction to the set of signs $\{-1,+1\}$ instead of all non-zero field elements in the proof of~\cref{cor:LCE=>PCE}.

\subsection{From \PCE to \LCE}
\label{subsec:PCE=>LCE}

First we show the forward direction.
We prove that our construction preserves the permutation equivalence of the input pair, which gives the following stronger result.

\begin{lemma}
    \label{lem:PCE=>PCE}
    For any $\mG, \mH \in \F_q^{k\times n}$, 
    if $(\mG, \mH)$ is in \emph{\PCE}, then $(\mG', \mH')$ is in \emph{\PCE} (and therefore in \emph{\LCE}).
\end{lemma}

\begin{proof}
    By definition, $\mG, \mH$ is in \PCE if and only if there exist an invertible matrix $\mS \in \GL_k$ and a permutation matrix $\mP \in \cP_n$ such that $\mS \mG \mP = \mH$.
    From these matrices we will construct $\mS' \in \GL_{k+1}$ and $\mP' \in \cP_{n'}$ such that $\mS' \mG' \mP' = \mH'$.
    Define $\mS' \in \GL_{k+1}$ such that the top-left block submatrix is $\mS'[1:k, 1:k] = \mS$, the bottom-right entry is $\mS'[k+1, k+1] = 1$, and all other entries are zero.

    Let $\sigma_\mP \colon [n] \to [n]$ be the permutation map described by $\mP$.
    Define a new permutation $\sigma_{\mP'} \colon [n'] \to [n']$ that permutes the first $n$ columns of $\mG'$ among each other according to $P$, i.e. $\sigma_{\mP'}(x) \defn \sigma_\mP(x)$ for all $x \in [n]$. 
    We also define $\sigma_{\mP'}$ such that it permutes each of the $nm$ columns in the second block according to how its corresponding column was permuted in the first block.
    Formally, for every $x \in [n+1, n+nm]$, we write $x = n + m(i-1) + j$ uniquely for $i \in [n]$ and $j \in [m]$ and set $\sigma_{\mP'}(x) \defn n+m(\sigma_\mP(i)-1)+j$.
    Lastly, we define $\sigma_{\mP'}$ such that it sends each of the last $nm+1$ columns identically to itself, i.e. $\sigma_{\mP'}(x) \defn x$ for every $x \in [n+nm+1, n+2nm+1]$.
    Finally we let $\mP' \in \cP_{n'}$ be the permutation matrix defined by $\sigma_{\mP'}$.

    We claim that $\mS' \mG' \mP' = \mH'$.
    This holds if and only if $\mS' \mG' \mP' [x] = \mH'[x]$ for every column $x\in [n']$.
    We show that this is the case for every column in each of the three blocks.
    For every $x \in [n]$, by the definition of $\sigma_{\mP'}$, the fact that each column in the first block has 1 as the last entry, and the \PCE assumption,
    \begin{equation*}
        \mS' \mG' \mP' [x] 
        = \mS' \mG'[\sigma_{\mP'}(x)]
        = \mS \mG[\sigma_\mP(x)]\ \|\ 1
        = \mS \mG[x] \mP\ \|\ 1
        = \mS \mG \mP [x]\ \|\ 1 
        = \mH [x]\ \|\ 1
        = \mH' [x].
    \end{equation*}
    Here ``$\|$" denotes concatenation of an additional entry at the end of the column vector.
    For $x \in [n+1, n+nm]$, we can write $x = n + m(i-1) + j$ for some $i \in [n]$ and $j \in [m]$.
    Then by the definition of $\sigma_{\mP'}$, the fact that each column in the second block has 0 as the last entry, and the \PCE assumption,
    \begin{align*}
        \mS' \mG' \mP' [x] 
        &= \mS \mG[\sigma_{\mP'}(x)]\ \|\ 0
        = \mS \mG[n+m(\sigma_\mP(i)-1)+j]\ \|\ 0
        = \mS \mG[\sigma_\mP(i)]\ \|\ 0 \\
        &= \mS \mG \mP [i]\ \|\ 0
        = \mH [i]\ \|\ 0
        = \mH [n+m(i-1)+j]\ \|\ 0
        = \mH' [x].
    \end{align*}
    Lastly, for all $x \in [n+nm+1,n+2nm+1]$, by definition of $\sigma_{\mP'}$, the fact that each column in the last block is identical to $\vece_{k+1}$, and the \PCE assumption,
    \begin{equation*}
        \mS' \mG' \mP' [x] 
        = \mS \mG[\sigma_{\mP'}(x)]\ \|\ 1
        = \mS \mG[x]\ \|\ 1
        = \mS \mG \mP [x]\ \|\ 1
        = \mH[x]\ \|\ 1
        = \mH'[x].
    \end{equation*}
    Therefore, $(\mG', \mH')$ is in \PCE.
\end{proof}

Because the matrix groups are nested $\cP_n \subseteq \SP_n \subseteq \cM_n$, \cref{lem:PCE=>PCE} immediately implies the forward direction of both statements~\cref{thm:main-PCE->LCE} and~\cref{thm:main-PCE->SPCE}.



\subsection{From \LCE to \PCE}
\label{subsec:LCE=>PCE}

Now we prove the other direction, and show that if the constructed matrix pair $(\mG', \mH')$ is in \LCE, then the original matrix pair $(\mG, \mH)$ must be in \PCE.
This requires some careful analysis of how the change-of-basis matrix $\mS$ and monomial matrix $\mM$ affect each block of the matrices $\mG'$ and $\mH'$.

First we show that for any monomial matrix $\mM = \mD\mP$ for which $(\mG', \mH')$ is in \LCE, the permutation matrix $\mP$ must respect the ``boundaries" between the block matrices of $\mG'$.

\begin{lemma}
    \label{lem:P'-structure}
    For any $\mG, \mH \in \F_q^{k\times n}$,
    if $\mS' \mG' \mM' = \mH'$ for some $\mS' \in \GL_{k+1}$ and $\mM'=\mD' \mP' \in \cM_{n'}$, where $\mP' \in \cP_{n'}$,
    the corresponding permutation map $\sigma_{\mP'} \colon [n']\to [n']$ satisfies the following:
    \begin{enumerate}[(i)]
        \item For every $i \in [1,n]$, $\sigma_{\mP'}(i) \in [1,n]$.
        \item For every $i \in [n+1,n+mn]$, $\sigma_{\mP'}(i) \in [n+1,n+mn]$.
        \item For every $i \in [n+nm+1,n']$, $\sigma_{\mP'}(i) \in [n+mn+1,n']$. 
    \end{enumerate}
\end{lemma}

\begin{proof}
    As a result of~\cref{lem:no-zero-column}, we can assume without loss of generality that $\mG$ and $\mH$ do not contain a zero column.
    Suppose $\mS' \in \GL_{k+1}$ and $\mM' = \mD' \mP' \in \cM_{n'}$ satisfy $\mS' \mG' \mM' = \mH'$.
    By definition, $\mS'$ has an inverse $\mS'^{-1} \in \GL_{k+1}$, so we can rewrite
    \begin{equation}
        \label{eq:LCE-equation}
         \mG' \mM' = \mS'^{-1} \mH' .
    \end{equation}
    By~\cref{rem:S-bijection}, $\mS'^{-1}$ maps identical columns in $\mH'$ to identical columns in $\mS'^{-1}\mH'$.
    To prove each part of the claim, we analyze the effect of $\mS'^{-1}$ on each block $\mH'_1, \mH'_2, \mH'_3$ of matrix $\mH'$.

    First we prove part (iii).
    By construction, the last $nm+1$ columns of $\mH'$, contained in $\mH'_3$ are identical to $\vece_{k+1}$.
    By assumption, $\mH$ does not contain a zero column, so no other column in $\mH'$ outside of the block $\mH'_3$ is equal to $\vece_{k+1}$.
    Then each of the $nm+1$ columns in $\mS'^{-1}\mH'_3$, and no other column of $\mS'^{-1}\mH'$, is equal to $\mS'^{-1}\vece_{k+1}$.

    By~\cref{eq:LCE-equation}, the last $nm+1$ columns of $\mG' \mM'$ must be identical and equal to $\mS'^{-1}\vece_{k+1}$.
    Since $\mM' = \mD' \mP'$ for some diagonal matrix $\mD'$ and permutation $\mP'$,
    it maps each column in $\mG'$ to a (possibly) scaled permutation of that column in $\mG' \mM'$.
    Thus, the matrix $\mM'$ must scale and permute the columns of $\mG'$ to produce $nm+1$ identical columns in the last block of $\mG' \mM'$.
    We claim that multiplying by $\mM'$ cannot cause the number of copies of a column in $\mG'$ to increase in $\mG' \mM'$.
    In particular, no two columns from different blocks of $\mG'$ can be scaled to produce identical columns in $\mG' \mM'$.
    By construction, the last row of $\mG'$ ensures that no column of $\mG'_1$ can be scaled to produce a column of $\mG'_2$, so $\mM'$ cannot map columns from $\mG'_1$ and $\mG'_2$ to identical columns in $\mG'\mM'$.
    Additionally, $\mM'$ maps every $\vece_{k+1}$ column in $\mG'_3$ to a scaling of $\vece_{k+1}$ in $\mG' \mM'$, 
    and since no column in $\mG'_1$ or $\mG'_2$ can be scaled to produce a multiple of $\vece_{k+1}$, all scalings of $\vece_{k+1}$ in $\mG' \mM'$ can only come from $\mG'_3$.
    The only column in $\mG'$ that appears $nm+1$ times is $\vece_{k+1}$ from $\mG'_3$.
    In this way, $\mM'$ can only produce $nm+1$ identical columns in the last block of $\mG' \mM'$ by mapping from the $nm+1$ columns of $\mG'_3$.
    Therefore, $\sigma_{\mP'}(i)$ maps every column with index in $[n+mn+1,n']$ to a column with index in $[n+mn+1,n']$.

    Next we show part (ii).
    By part (iii), any column with index $i \in [n+1,n+mn]$ is permuted to a column with index $\sigma_{\mP'}(i) \leq n+mn$, so it is enough to show that $\sigma_{\mP'}(i) \geq n+1$.
    Suppose for the sake of contradiction that there is a column with index $i \in [n+1,n+nm]$ in $\mG'$ that is mapped to a column with index $\sigma_{\mP'}(i) \in [1,n]$.
    By design, the last entry of the columns in $\mG'_1$ is 1 and differs from the last entry 0 of the columns in $\mG'_2$, so each column in $\mG'_1$ appears strictly less than $m$ times.
    Then, if $\mM'$ permutes any number of columns from $\mG'_1$ with the same number of columns in $\mG'_2$, there will be some column in the second block of $\mG' \mM'$ that appears less than $m$ times.
    But by~\cref{rem:S-bijection}, every column in $\mS'^{-1} \mH'_2$ appears at least $m$ times, and since the second block of $\mG' \mP'$ must be equal to $\mS'^{-1} \mH'_2$, this gives a contradiction.

    Finally, part (i) follows immediately from parts (ii) and (iii).
\end{proof}

This result gives the following corollary.
\begin{corollary}
    \label{cor:M'-structure}
    For any $\mG', \mH', \mS', \mM'$ as in~\cref{lem:P'-structure}, the monomial matrix $\mM'$
    is comprised of three block matrices $\mM_1 \in \cM_n$, $\mM_2 \in \cM_{nm}$, $\mM_3 \in \cM_{nm+1}$, such that $\mM_1$, $\mM_2$, and $\mM_3$ only act on the first $n$, next $nm$, and last $nm+1$ columns of $\mG'$, respectively.
\end{corollary}

\begin{figure}[H]
    \centering
    \begin{tikzpicture}
        \draw[-] (0,0) -- (4,0) node[right] {};
        \draw[-] (0,4) -- (4,4) node[right] {};
        \draw[-] (0,0) -- (0,4) node[right] {};
        \draw[-] (4,0) -- (4,4) node[right] {};
        \coordinate (w1) at (1,4);
        \coordinate (h1) at (0,3);
        \coordinate (o2) at (1,3);
        \draw[-] (w1) -- (o2) node[right] {};
        \draw[-] (h1) -- (o2) node[right] {};
        \coordinate (w2) at (2.25,3);
        \coordinate (h2) at (1,1.75);
        \coordinate (o3) at (2.25,1.75);
        \draw[-] (o2) -- (w2) node[right] {};
        \draw[-] (o2) -- (h2) node[right] {};
        \draw[-] (w2) -- (o3) node[right] {};
        \draw[-] (h2) -- (o3) node[right] {};
        \coordinate (w3) at (4,1.75);
        \coordinate (h3) at (2.25,0);
        \draw[-] (o3) -- (w3) node[right] {};
        \draw[-] (o3) -- (h3) node[right] {};
        \node[label={$\mM_1$}] at (0.5,3.05) {};
        \node[label={$\mM_2$}] at (1.6,1.9) {};
        \node[label={$\mM_3$}] at (3.1,0.4) {};
    \end{tikzpicture}
    \caption{The structure of matrix $\mM'$.}
    \label{fig:M'-structure}
\end{figure}

Now, we use this lemma to show that any invertible matrix for which $(\mG', \mH')$ is in \LCE must contain only zeros in the last row and last column, except for the last entry.
\begin{lemma}
    \label{lem:S'-structure}
    For any $\mG', \mH', \mS', \mM'$ as in~\cref{lem:P'-structure}, 
    the change-of-basis matrix $\mS'$ satisfies the following properties:
    \begin{enumerate}[(i)]
        \item The last column of $\mS'$ contains zeros in the first $k$ entries.
        \item The last row of $\mS'$ contains zeros in the first $k$ entries.
        \item The entry $\mS'[k+1,k+1]$ is non-zero.
    \end{enumerate}
\end{lemma}

\begin{figure}[H]
    \centering
    \begin{tikzpicture}
        \draw[-] (0,0) -- (3,0) node[right] {};
        \draw[-] (0,3) -- (3,3) node[right] {};
        \draw[-] (0,0) -- (0,3) node[right] {};
        \draw[-] (3,0) -- (3,3) node[right] {};
        \coordinate (w1) at (2.5,3);
        \coordinate (h1) at (0,0.5);
        \coordinate (o2) at (2.5,0.5);
        \draw[-] (w1) -- (o2) node[right] {};
        \draw[-] (h1) -- (o2) node[right] {};
        \coordinate (w2) at (3,0.5);
        \coordinate (h2) at (2.5,0);
        \draw[-] (o2) -- (w2) node[right] {};
        \draw[-] (o2) -- (h2) node[right] {};
        \node[label={$\mS$}] at (1.25,1.25) {};
        \node[label={$a$}] at (2.75,-0.1) {};
    \end{tikzpicture}
    \caption{The structure of matrix $\mS'$.}
    \label{fig:S'-structure}
\end{figure}

\begin{proof}
    By definition, $\mS'$ has an inverse $\mS'^{-1} \in \GL_{k+1}$, so we can write $\mG' \mM' = \mS'^{-1} \mH'$.
    By~\cref{rem:S-bijection}, $\mS'^{-1}$ maps identical columns in $\mH'$ to identical columns in $\mS'^{-1}\mH'$.
    To prove each part of the claim, we analyze the effect of $\mS'^{-1}$ on each block $\mH'_2$ and $\mH'_3$ in the right-hand side of the equation.

    By~\cref{cor:M'-structure}, we know that $\mM'$ is comprised of three block matrices $\mM_1 \in \cM_n$, $\mM_2 \in \cM_{nm}$, and $\mM_3 \in \cM_{nm+1}$ that affect the first $n$, next $nm$, and last $nm+1$ columns of $\mG'$, respectively.
    This allows us to write $\mG'_2 \mM_2 = \mS'^{-1} \mH'_2$ and $\mG'_3 \mM_3 = \mS'^{-1} \mH'_3$.

    For part (i), consider the equation $\mG'_3 \mM_3 = \mS'^{-1} \mH'_3$.
    By construction, all columns of $\mG'_3$ and $\mH'_3$ are identical and equal to $\vece_{k+1}$.
    Since $\mM_3$ permutes all columns of $\mG'_3$ and multiplies them by a non-zero scalar, all columns of $\mG'_3 \mM_3$ must be of the form $a\cdot \vece_{k+1}$ for some non-zero $a$.
    By~\cref{rem:S-bijection}, all columns of $\mS'^{-1} \mH'_3$ are identical.
    Then for any column $a\cdot \vece_{k+1}$ of $\mG'_3 \mM_3$, we have $\mS'(a\cdot \vece_{k+1}) = \vece_{k+1}$.
    If any of the first $k$ entries of the last column of $\mS'$ is non-zero, then the corresponding entry in $\mS'(a\cdot \vece_{k+1}) = \vece_{k+1}$ must be non-zero, but this is not the case.
    Therefore, the last column of $\mS'$ must contain zeros in the first $k$ entries.

    For part (ii), consider the equation $\mG'_2 \mM_2 = \mS'^{-1} \mH'_2$.
    By construction, the last rows of $\mG'_2$ and $\mH'_2$ only contains zeros.
    By~\cref{lem:P'-structure}, no columns in $\mG'_2 \mM_2$ could have been mapped from outside $\mG'_2$, so the last row of $\mG'_2 \mM_2$ must only contain zeros.    
    Since $\mG$ is a submatrix of $\mG'_2$ with full row rank $k$, the first $k$ rows of $\mG'_2$, denoted by $\widehat{\mG}$ in~\cref{con:our-construction}, form a submatrix of rank $k$.
    Then $\widehat{\mG}$ also has column rank $k$, and because the last row of $\mG'_2$ only contains zeros, $\mG'_2$ must contain $k$ linearly independent columns.
    By~\cref{cor:M'-structure}, we know that $\mG'_2 \mM_2$ is entirely comprised of columns of $\mG'_2$ multiplied by a non-zero scalar, so $\mG'_2 \mM_2$ has column rank $k$.
    Then the first $k$ rows of $\mG'_2 \mM_2$ form a submatrix of rank $k$. 
    By definition, the last row of $\mS'$ contains the coefficients that specify the linear combination of rows of $\mG'_2 \mM_2$ which gives the last row of $\mH_2$.
    But since the first $k$ rows of $\mG'_2 \mM_2$ are linearly independent, no linear combination of these can produce the all-zero last row of $\mH_2$.
    Therefore, the first $k$ entries in the last row of $\mS'$ must be zero.

    Part (iii) follows immediately from parts (i) and (ii) and the fact that $\mS'$ is invertible.
\end{proof}

Finally, we show how~\cref{lem:P'-structure,lem:S'-structure} combine to ensure the existence of an unsigned permutation for which $(\mG, \mH)$ is in \PCE.

\begin{corollary}
    \label{cor:M1-unsigned-permutation}
    Let $\mM' \in \cM_n$ be a permutation with block submatrices $\mM_1, \mM_2, \mM_3$ as in~\cref{cor:M'-structure}.
    Then, $\mM_1 = a \cdot \mP$ for some permutation $\mP \in \cP_n$ and non-zero scalar $a$.
\end{corollary}

\begin{proof}
    By~\cref{cor:M'-structure}, we can write $\mG'_1 \mM_1 = \mS'^{-1} \mH'_1$.
    By~\cref{lem:S'-structure} and~\cref{rem:S-bijection}, the last row of $\mS'^{-1} \mH'_1$, and hence the last row of $\mG'_1 \mM_1$, must be of the form $(a,a,\dots,a)$ for some non-zero scalar $a$.
    By construction, the last row of $\mG'_1$ contains only ones. 
    Since $\mM_1$ acts on $\mG'_1$ by permuting and scaling the columns of $\mG'_1$, and since the last row of $\mG'_1 \mM_1$ contains identical entries $a$, we infer that $\mM_1$ must multiply each column by the same scalar $a$.
    Therefore, $\mM_1 = a\cdot \mP$ for some unsigned permutation $\mP \in \cP_n$ and non-zero scalar $a$.
\end{proof}

Finally, we use the structure of the permutation and change-of-basis matrix for any \LCE matrices $\mG'$ and $\mH'$ to show that the original matrices $\mG$ and $\mH$ must be in \PCE.

\begin{corollary}
    \label{cor:LCE=>PCE}
    For any $\mG, \mH \in \F_q^{k\times n}$, 
    if $(\mG', \mH')$ is in \emph{\LCE}, then $(\mG, \mH)$ is in \emph{\PCE}.
\end{corollary}

\begin{proof}
    If $(\mG', \mH')$ is in \LCE, then there exists a monomial matrix $\mM' \in \cM_{n'}$ and an invertible matrix $\mS' \in \GL_{k+1}$ such that $\mS' \mG' \mM' = \mH'$.

    By~\cref{cor:M'-structure}, we know that $\mM'$ is comprised of three block matrices $\mM_1 \in \cM_n$, $\mM_2 \in \cM_{nm}$, and $\mM_3 \in \cM_{nm+1}$ which act exclusively on the first $n$, next $nm$, and last $nm+1$ columns of $\mG'$, respectively.
    By~\cref{cor:M1-unsigned-permutation}, we know that $\mM_1 = a \cdot \mP$ for some unsigned permutation $\mP \in \cP_n$ and non-zero scalar $a$.

    By~\cref{lem:S'-structure}, the last row and last column of $\mS'$ contain only zeros, except in the last entry.
    Let $\mS \in \GL_k$ denote the top-left block submatrix of $\mS'$ consisting of the intersection of the first $k$ rows and $k$ columns (see~\cref{fig:S'-structure}).
    Since $\mS'$ must have a non-zero determinant, this implies that $\mS$ must be invertible.
    Then since $a$ is non-zero, $a \cdot \mS$ is also an invertible matrix.
    
    By construction of $\mG'$ and $\mH'$, \cref{lem:S'-structure}, and~\cref{cor:M'-structure}, we have the block matrix product $\mS \mG \mM_1 = \mH$.
    Then for the matrices $(a \cdot \mS)$ and $\mP$ we obtain
    \begin{equation*}
        (a \cdot \mS) \mG \mP = \mS\ \mG\ (a \cdot \mP) = \mS \mG \mM_1 = \mH.
    \end{equation*}
    Therefore, $(\mG, \mH)$ is in \PCE.
\end{proof}

For the reverse direction of~\cref{thm:main-PCE->SPCE}, the proof is nearly identical to the one for \LCE, but with the assumption that all non-zero scalars are restricted to the set of signs $\{-1,+1\}$.

\section{Conclusion}
\label{sec:conclusion}

Our results imply that the \CE variants described above are nearly computationally equivalent, from the perspective of polynomial-time algorithms. 
In particular, we have shown that \LCE and \SPCE are at least as hard as \PCE.
So, in order to study the hardness of the former two problems, it suffices to focus on the hardness of \PCE.
We now have more reductions among the \CE variants, and as a consequence of~\cite{BW24}, from these variants to \LIP, yet it remains an open problem to reduce \LIP to any \CE problem.
While there has been some small progress in this direction, namely Ducas and Gibbons' Turing reduction from \LIP for Construction A lattices to \SPCE for codes with trivial hull~\cite{DG23}, no other progress has yet been made.

Another interesting open problem is improving the runtime of the reduction from \LCE to \PCE given in~\cite{SS13a} from $\poly(n, q)$ to $\poly(n, \log q)$. 
This reduction requires time polynomial in the alphabet size $q$ because the closure (as described in~\cref{sec:intro}) increases the blocklength of the codes from $n$ to $(q-1)n$. 
Any reduction running in $\log q$ time would either need to use a new type of closure that requires only a $\log q$ increase in blocklength, or a more efficient way to transform scalar multiplication operations into permutations.

\section{Acknowledgments}
\label{sec:acknowledge}

The authors would like to thank Huck Bennett, Chris Peikert, and Yi Tang for helpful conversations during the course of this work.
This research was partially supported by the National Science Foundation under Grant No.\ CCF-2236931.

\newpage

\bibliographystyle{alpha}
\bibliography{references.bib}

\end{document}